%% file: main.tex
\documentclass[12pt]{article}
\usepackage[margin=1in]{geometry}
\usepackage{hyperref}
\usepackage{bbding}

\usepackage[english]{babel}
\usepackage{amsmath,amsthm,bm}
\usepackage{amssymb}
\usepackage{amsfonts}

\usepackage{algorithm}  
\usepackage{algorithmicx}  
\usepackage{algpseudocode}  

\usepackage{color}

\usepackage{CJKutf8}
\usepackage{graphicx}
\usepackage{pgf,tikz}
\usepackage{enumitem}
\usetikzlibrary{shapes,arrows,automata}

\input{mycmd}

\begin{document}
\title{A Top-Down Depth-Four Circuit Lower Bound for Majority}%

\author{Hao Wu\footnote{Corresponding author. College of Information Engineering, Shanghai Maritime University, Shanghai, China. My email is \texttt{haowu@shmtu.edu.cn}, you can also reach me via \texttt{ wealk@outlook.com}.}
\and Yaqiao Li\footnote{Shenzhen University of Advanced Technology, Shenzhen, China, liyaqiao@suat-sz.edu.cn}
}%

\maketitle

\begin{abstract}
We present a top-down depth-four circuit lower bound for Majority function by extending recent work of G\"{o}\"{o}s, Riazanov, Sofronova, and Sokolov (FOCS 2023), who gave a top-down proof of a depth-four circuit lower bound for Parity which relies on the robust sunflower to construct a mirror set and the block unpredictability to find the local limits. The main challenge for the case of Majority is to construct a corresponding mirror set, the  difference is that to flip the value of Majority function, one may have to flip many bits of the input Boolean string, while for Parity, flipping one bit suffices. We avoid this flipping by considering slices of the Boolean cube, that is, Boolean strings of fixed Hamming weight approximately $n/2$.
\end{abstract}


\section{Introduction}
Proving strong circuit lower bounds remains one of the central challenges in complexity theory. One of the most successful methods is the random restriction which analyzes the circuit in a bottom-up way. With this method, a celebrated line of works \cite{DBLP:journals/apal/Ajtai83,DBLP:journals/mst/FurstSS84,10.5555/4479.4487, has87,haastad2014correlation} show that to compute functions of $n$ variables such as Parity or Majority requires a depth-$d$ unbounded fan-in circuit of size  at least $2^{\Omega(n^{1/(d-1)})}$. 

Another major methodology to prove circuit lower bounds is the top-down method. To prove a  Boolean function $f$ requires a large  circuit to compute, the top-down  approach works as follows. Assume for the sake of a contradiction that there exists a small circuit for function $f$, starting from the top of the circuit, we  choose a sub-circuit and then iterate the process down to the bottom. Note that every sub-circuit computes a partial function consistent with the original function $f$, in other word, every sub circuit has to distinguish some set  $X\subseteq f^{-1}(1)$ from another set $Y\subseteq f^{-1}(0)$. Along the way, eventually we will reach a contradiction  that certain sub-circuit fails to  distinguish its $X$ from $Y$. 

The merit of the top-down approach is that it is complete for proving constant depth circuit lower bound and has the potential to prove circuit lower bound with modular gates, see \cite{hirahara2017duality,goos2023depth,DBLP:conf/innovations/RiazanovS026}. The difficulty of this approach is that we choose the sub-circuit in an ``online'' fashion while ignoring any structure information of lower layers. For many years, we only have top-down depth-three lower bounds. For example,  H\aa{}stad, Jukna, and Pudl\'{a}k \cite{haastad1995top} demonstrated a top-down method to prove that any depth-three circuit computing $n$-bit Parity or Majority requires size $2^{\Omega(\sqrt{n})}$. Meir and Wigderson \cite{DBLP:journals/cc/MeirW19} provided another top-down proof for this result by inventing an information theoretical unpredictability lemma. 

Recently, a breakthrough was achieved by G\"{o}\"{o}s, Riazanov, Sofronova, and Sokolov \cite{goos2023top}, who established the first top-down proof of the $2^{n^{1/3-o(1)}}$ size lower bound for depth-four boolean circuits computing the $n$-bit Parity function.  G\"{o}\"{o}s et al.\ overcame this long-standing barrier by introducing mainly two new components: the construction of a mirror set, and a novel block unpredictability lemma that generalizes the work of Meir and Wigderson \cite{DBLP:journals/cc/MeirW19}. However,  their lower bound only works for the case of Parity,  it is natural to adapt their method for other cases; a good candidate is the Majority function. Below we list the main steps of their proof for Parity, and then explain the challenge to extend it to the case of Majority.

\begin{description}
	\item[Step one:] Assume that we have an $\wedge$-gate on top. We greedily choose  a sub-circuit which has to distinguish  $\oplus^{-1}(1)$ from $Y\subseteq \oplus^{-1}(0)$ where the size of $Y$ is maximized among all the sub-circuits in the same layer. 
		\item[Step two:] Construct a mirror set $M \subseteq \oplus^{-1}(1)$ where for every $x$ in the mirror set $M$, if we flip some bits in a random block $R\subseteq [n],|R|=r$, then with overwhelming probability, it turns $x$ into some  $y\in Y$. Then we greedily choose  a sub-circuit which has to distinguish  $M'\subseteq M$ from $Y$ where the size of $M'$ is maximized.
		\item[Step three:] Construction of the set of $q$-local limits $Y' \subseteq Y$ where for every $y\in Y'$, if we only query $q$ bits in $y$, there exists an $x\in M'$ such that if we query the same $q$ coordinates of $x$, they are the same. Intuitively, we cannot distinguish $y$ from $M'$ locally. Then we  greedily choose  a sub-circuit which has to distinguish  $M'$ from $Y''\subseteq Y'$ where the size of $Y''$ is maximized.
		\item[Reaching a contradiction:] After the three steps, the size of $Y''$  is large enough and  we obtain an $\vee$-clause of (at most) $q$ literals which has to  distinguish  $M'$ from $Y''\subseteq Y'$, which is a contradiction since every $y\in Y''$ is a $q$-local limit for $M'$.
	\end{description}
    
The main challenge to prove a lower bound for  the Majority function $\texttt{Maj}$ is to construct the corresponding mirror set. In the case of Parity function $\oplus$,  to find an element $x'$ in the mirror set $M\subseteq \oplus^{-1}(1)$, firstly, one has to find certain point $x\in \oplus^{-1}(1)$ and shift $x$ to $x''$ by adding some shift vector $a$, if $x''$ is not in $\oplus^{-1}(1)$, we have to flip one bit in $x''$ to obtain $x'$, in order to make sure that $x'$ is in $\oplus^{-1}(1)$. This argument does not work for the Majority function because to  flip the value of Majority function, one may have to flip many bits of the string $x''$,  while for Parity,  flipping one bit suffices. We avoid this flipping by considering Boolean strings of fixed Hamming weight slightly larger than $n/2$,  and make sure the  Hamming weight of the shift vector $a$ is small so that after shifting, it is still in $\texttt{Maj}^{-1}(1)$ thus avoiding the flipping. Formally, we have the following top down depth-four unbounded fan-in circuit lower bound for the Majority function $\texttt{Maj}$.

\begin{theorem}\label{thm1.1}
    Every depth-$4$  circuit computing the $n$-bit Majority function  requires at least $2^{n^{1/3-o(1)}}$ gates.
\end{theorem}

Our result can be easily extended to the case of threshold functions with linear threshold, and also the case of slice function.
The rest of the paper is organized as follows. We provide basic notations and necessary tools in Section ~\ref{section2}. We prove our main result Theorem \ref{thm1.1} in Section~\ref{section3}. Finally, in Section~\ref{section4} we conclude and make some discussion.

\section{Notations and tools}\label{section2}

At first, we introduce some basic notations. Let $[n]$ denote the set $\{1,\ldots,n\}$, $2^{[n]}$ the power set of $[n]$. Denote $\{S\subseteq [n] \mid |S| =r\}$ by $\binom{[n]}{r}$. Given a  Boolean vector $x \in \{0,1\}^n$, define $I_x$ to be the set $\{i|x_i =1\}$. Conversely, we say $x$ is the indicator vector for $I_x$. Sometimes for convenience, we will use $x$ and $I_x$ interchangeably. Given a  Boolean vector $x \in \{0,1\}^n$ and a set of coordinates $I\subseteq [n]$ , denote $x$ restricted to $I$ by $x_I$. Given a set $X\subseteq \{0,1\}^n$, denote the set $\{x_I\mid x\in X\}$ by $X_I$. We often view $x \in \{0,1\}^n$ as a vector of $\mathbb{F}_2^n$, the Hamming weight of $x$ denoted by  $\|x\|_0$ is the number of ones in $x$, and Hamming distance between two vectors $x,y\in \{0,1\}^n$ is denoted by $\|x-y\|_0$.  Given a  Boolean vector $x \in \{0,1\}^n$ and a set $Y\subseteq \{0,1\}^n$, $Y-x$ is the set $\{y-x\mid y\in Y\}$. Now we introduce some key definitions and tools.

\begin{defn}[$(p, \varepsilon)$-satisfying \cite{Ross14,alwz21,rao20}]Given a set family $\mathcal{A}\subseteq 2^{[n]}$, $p
	=r/n, r\in [n],0<\varepsilon<1$, we say $\mathcal{A}$ is  $(p, \varepsilon)$-satisfying, if
	\[
	\Pr_{\bm{R}\sim \binom{[n]}{pn} }[\exists A\in \mathcal{A}: A\subseteq \bm{R}]\geq 1-\varepsilon.
	\]
\end{defn}

\begin{defn}[Spreadness\cite{alwz21,rao20}]Given a set family $\mathcal{A}\subseteq 2^{[n]}$, we say $\mathcal{A}$ is  $\kappa$-spread, if for every $I\subseteq [n]$
	\[\Pr_{\bm{A}\sim \mathcal{A}}[I\subseteq \bm{A}]\leq \kappa^{-|I|}.\]
\end{defn}

The robust sunflower lemma says that  spreadness implies satisfying.

\begin{lemma}[\cite{rao20}]\label{lem2.3}There exists a universal constant $C > 0$ such that the following holds. Given a set family $\mathcal{A}\subseteq \binom{[n]}{r}$, if $\mathcal{A}$ is  $\kappa$-spread for $\kappa = (C/p) \log(r/\varepsilon)$. Then $\mathcal{A}$ is $(p, \varepsilon)$-satisfying.
\end{lemma}

The robust sunflower lemma will be used to construct the mirror set. Now we define the notion of local limit. Intuitively, it means we cannot distinguish a string from a set of strings if we only locally query a limited number of coordinates.

 \begin{defn}[$q$-local limit\cite{goos2023top}]Given $X\subseteq \{0,1\}^n, y\in \{0,1\}^n$, we say $y$ is a $q$-local limit for $X$ if for every $Q\subseteq [n],|Q|=q$, there exists some $x\in X$ such that $x_Q =y_Q$. Intuitively, we cannot distinguish $y$ from $X$ if we only locally query $q$ coordinates in $y$.
 \end{defn}
 
 Now we define the certificate for block of coordinates and introduce the block unpredictability lemma.
 
 \begin{defn}[\cite{goos2023top}] Let  $X \subseteq \{0,1\}^n$ be a set of Boolean vectors and $R\in \binom{[n]}{r}$ be a block of coordinates. Given a pair $(Q, a)$ such that $Q \subseteq [n]\setminus R, a \in \{0, 1\}^Q$, let $X_{Q=a} :=\{x\in X\mid x_Q =a\}$, we say that $(Q, a)$ is a certificate for $R$
  with respect to $X$ if $(X_{Q=a})_R \neq \{0,1\}^R$. Intuitively, the certificate $(Q, a)$ predicts that some string in $\{0,1\}^R$ is missing in the set $(X_{Q=a})_R$.
 Furthermore, we say $x \in X$ contains a size-$q$ certificate $(Q,a)$ for $R$ w.r.t. $X$, if $(Q,a)$ is a certificate of  $R$ such that $|Q|=q,x_Q =a$.
 \end{defn}
 
\begin{lemma}[\cite{goos2023top}]\label{lem2.6}Let $X \subseteq \{0,1\}^n$, and $|X|\geq 2^{n-k}$. For every $r,q\geq 1$, we have
\[
	\Pr_{(\bm{x,R})\sim X\times \binom{[n]}{r}} \left[ \text{ $\bm{x}$ contains a size-$q$ certificate  for $\bm{R}$ w.r.t. $X$}\right]\leq  O(kqr/n)^{1/6}.
	\]
In particular, there is a subset $X'\subseteq X,|X'|\geq |X|/2$, and for every $x\in X'$,
\[
\Pr_{\bm{R}\sim\binom{[n]}{r}} \left[x \text{ contains a size-$q$ certificate  for } \bm{R} \text{ w.r.t. } X\right]\leq  O(kqr/n)^{1/6}.
\]
\end{lemma}

The merit of block unpredictability is that it implies existence of local limit for the mirror set. The following lemma is implicit in \cite{goos2023top}, for completeness of this paper, we provide the proof here.

\begin{lemma}[\cite{goos2023top}]\label{lem2.7}Let $X \subseteq \{0,1\}^n, x\in X, R\subseteq [n]$ and $x$ does not contain any size-$q$ certificate  for $R$ w.r.t. $X$. If $y$ is obtained from $x$ by flipping some bits of $x$ in $R$, that is $y_{[n]\setminus R} = x_{[n]\setminus R}$, then $y$ is a $q$-local limit for $X$.
\end{lemma}

\begin{proof}
It suffices to show for every $Q\subseteq [n],|Q|=q$, we have $y_Q =z_Q$ for some $z\in X$. Let $Q' = Q\setminus R$, thus $|Q'|\leq q$, note that $y_{Q'} =x_{Q'}$. Since $x$ does not contain any size-$q$ certificate  for $R$ w.r.t. $X$, that is  $(X_{Q'=x_{Q'}})_R = \{0,1\}^R$, this means there is a $z\in X_{Q'=x_{Q'}}$ such that $z_R =y_R$ and $z_{Q'} =x_{Q'}$, the latter implies $z_{Q'} =x_{Q'}=y_{Q'}$, therefore, $z_{Q'\cup R} = y_{Q'\cup R}$ ,  that is $z_{Q\cup R} = y_{Q\cup R}$ as required.
\end{proof}

\section{Main result}\label{section3}
In this section, we prove our main result Theorem \ref{thm1.1} with top down approach. The proof strategy is similar to that in \cite{goos2023top}, the main difference is the construction of the mirror set. At first, recall that for any $x \in \{0,1\}^n$, the Majority function $\texttt{Maj}$ is defined to be $\texttt{Maj}(x)=1$ if and only if $\|x\|_0> \lfloor n/2 \rfloor$.
Let $m := n^{1/3}$, $\delta$ be some parameter such that  $\delta =o(1)$ and meanwhile $\delta=\omega(\frac{\log\log m}{\log m})$. Suppose for the sake of contradiction that $\Pi$ is a depth-4 circuit of size $|\Pi| \leq 2^{m^{1-\delta}}$ that computes the $n$-bit Majority function. W.l.o.g, assume that $\Pi$ is of type $\wedge \circ \vee \circ \wedge \circ \vee$, that is, it has an $\wedge$ gate at
the top,  rest layers alternate between $\vee$ and $\wedge$,  the bottom layer contains input literals. Starting at the top gate, we will take three steps down the circuit to reach the contradiction required.

\paragraph{Step one}~\\
Let the depth-4 circuit $\Pi$ be $\bigwedge_{i=1}^{s} \Sigma_i$ where every $\Sigma_i$ is a depth-3 circuit of type $\vee \circ \wedge \circ \vee$ and $s\leq  2^{m^{1-\delta}}$ is the fan-in of  top gate of $\Pi$. 
At first, define 
\[Y_0:=\{y \in \{0,1\}^n\mid\|y\|_0 =\lfloor n/2 \rfloor \}\subseteq \texttt{Maj}^{-1}(0)
\]
where $\|y\|_0$ is the Hamming weight of $y$. Thus we have
$ \texttt{Maj}^{-1}(1) \subseteq \Sigma_i^{-1}(1)$ and $Y_0 \subseteq \bigcup_{i=1}^{s} \Sigma_i^{-1}(0)$, now choose $\Sigma_i$  greedily so that $\Sigma_i$ maximizes  $Y_0 \cap \Sigma_i^{-1}(0)$. After this, we set $Y_1 = Y_0 \cap \Sigma_i^{-1}(0)$ and $\Sigma:=\Sigma_i$. To summarize,

\begin{itemize}
	\item $\Sigma$ accepts the set $\texttt{Maj}^{-1}(1)$, in particular,  it accepts the set $X_1 \subseteq \texttt{Maj}^{-1}(1)$ defined as follows,
	\[
	 X_1 :=\{x \in \{0,1\}^n\mid\|x\|_0 = \lfloor n/2 \rfloor +m\} .
	\]
	The size of $X_1$ is $2^n \cdot 2^{-\Theta(\log m)}$.
	\item  $\Sigma$  rejects the set $Y_1 = Y_0 \cap \Sigma^{-1}(0)$  such that $Y_1$ is a subset of $Y_0$  and $|Y_1|/|Y_0|\geq 1/s \geq 2^{-m^{1-\delta}}$.
\end{itemize}

\paragraph{Step two}~\\
At first, we construct a mirror set $M \subseteq \texttt{Maj}^{-1}(1)$ for  $Y_1$.
Intuitively, the mirror set $M$ is some set  ``between'' $X_1$ and $Y_1$ in the sense that for every $x' \in M$, its Hamming weight $\|x'\|_0$ satisfies $\lfloor n/2 \rfloor <\|x'\|_0 \leq \lfloor n/2 \rfloor +m$, and  $x'$ is obtained by choosing certain proper $x \in X_1$ then flipping a few ones in  $x$ to zeros, the number of the flipped ones is small so that the Hamming weight of $x'$ is still strictly larger than $n/2$. 
Formally, we have following lemma.

\begin{lemma}\label{lem3.1}There exists a subset $M \subseteq \texttt{Maj}^{-1}(1)$ such that for every $x' \in M$,  $Y_1-x'$ is $(m^{1+\Theta(\log\log m/\log m)}/n, o(1))$-satisfying. Furthermore, the size of  $M$ is at least $2^n\cdot 2^{-m^{1-\Theta(\delta)}}$.
\end{lemma}

We defer the proof of this lemma to Section \ref{section3.1}. Now we are  ready to  take the second step down the circuit. Recall  $\Sigma$ is a depth-3 circuit  of type   $\vee \circ \wedge \circ \vee$, furthermore it accepts the set $\texttt{Maj}^{-1}(1)$ and in particular,  the set $M$, rejects  the set $Y_1$. Suppose the depth-3 circuit $\Sigma$ is of form $\bigvee_{i=1}^{s} \Gamma_i$ where every $\Gamma_i$ is a depth-2 circuit of type $\wedge \circ \vee$. Thus, we have
$ M \subseteq \bigcup_{i=1}^{s}\Gamma_i^{-1}(1)$, and $Y_1 \subseteq \Gamma_i^{-1}(0)$ for every $i$. 
Now, we choose $\Gamma_i$ so that $\Gamma_i$ maximizes  $M \cap \Gamma_i^{-1}(1)$. After this , we set $X_2 = M \cap \Gamma_i^{-1}(1)$ and $\Gamma:=\Gamma_i$. To summarize, $\Gamma$ accepts the set $X_2 \subseteq M$, rejects the set $Y_1$. Furthermore, the size of $X_2$ is at least $|M|/2^{m^{1-\delta}} =2^n\cdot 2^{-m^{1-\Theta(\delta)}}$.

\paragraph{Step three}~\\
Before going further, we want to setup some parameters properly. At first, set $k$ to ensure $|X_2|= 2^n \cdot 2^{-k}$, thus $k = m^{1-\Theta(\delta)}$. To invoke Lemma \ref{lem3.1}, we set $r$ to
ensure $r/n \geq m^{1+\Theta(\frac{\log\log m}{\log m})}/n$, thus $r = m^{1+\Theta(\frac{\log\log m}{\log m})}$. At the end of Step three, we will construct certain set $Y_3$ with size (at least) $2^n \cdot 2^{-m^{1+\Theta(\frac{\log\log m}{\log m})}}$, and we set $q$ to ensure $Y_3\geq 2^n\cdot 2^{-q}$, thus $q=m^{1+\Theta(\frac{\log\log m}{\log m})}$. Now apply Lemma \ref{lem2.6} to the set $X_2$ with these parameters
$k, r, q $, since $krq \leq  o(n)$, we have following fact.

\begin{claim}\label{claim3.2}There is a subset $X_3\subseteq X_2,|X_3|\geq |X_2|/2$, and for every $x\in X_3$,
	\[
	\Pr_{\bm{R}\sim\binom{[n]}{r}} \left[x \text{ contains a size-$q$ certificate  for } \bm{R} \text{ w.r.t. } X_2\right]\leq  o(1).
	\]
\end{claim}

Now we show how to sample many local limit $y$ in $Y_1$ according to the element $x$ in the set $X_3$.  Given any $x\in X_3$, use following process to sample a random $\bm{y}(x) \in Y_1\cup \{\bot\}$ where $\bot$ means failure.

\begin{enumerate}[label=(\arabic*)]
	\item Sample  a random $\bm{R}\sim \binom{[n]}{r}$.
	\item \label{step2}If $x$ contains a size-$q$ certificate for $\bm{R}$ w.r.t $X_2$, output $\bm{y}(x):=\bot$ .
	\item \label{step3}If there is some $y \in Y_1$ such that $y$ is obtained by flipping some bits in $\bm{R}$ of $x$, in other words, $y_{[n]\setminus \bm{R}}=x_{[n]\setminus \bm{R}}$, output $\bm{y}(x):=y$. Otherwise output $\bm{y}(x):=\bot$.
\end{enumerate}

\begin{claim}\label{claim3.3}Let $Y_2:= \bigcup_{x \in X_3}\text{supp}(\bm{y}(x)) \setminus\{\bot\} \subseteq Y_1$, then we have every $y\in Y_2$ is a $q$-local limit  for $X_2$ and the size of $Y_2$ is at least $2^n \cdot 2^{-m^{1+\Theta(\frac{\log\log m}{\log m})}}$.
\end{claim}

\begin{proof}At first  by Lemma \ref{lem2.7}, we know every $y\in Y_2$ is a $q$-local limit  for $X_2$. Now we estimate the number of distinct $y$'s generated by this process. Given any $x\in X_3$, the probability to successfully obtain a corresponding $y$ is $1-o(1)$, since by Claim \ref{claim3.2} step \ref{step2} fails with $o(1)$ probability and by Lemma \ref{lem3.1} step \ref{step3} fails with $o(1)$ probability. And in this process, the multiplicity of $y$ is at most $\sum_{i=0}^{r} \binom{n}{i} =2^{m^{1+\Theta(\frac{\log\log m}{\log m})}}$. Thus the size of $Y_2$ is at least $|X_3|\cdot (1-o(1))/ 2^{m^{1+\Theta(\frac{\log\log m}{\log m})}} = 2^n \cdot2^{- m^{1-\Theta(\delta)}}\cdot 2^{-m^{1+\Theta(\frac{\log\log m}{\log m})}}=2^n \cdot 2^{-m^{1+\Theta(\frac{\log\log m}{\log m})}}$.
\end{proof}

It is time to take our third thus final step down the circuit. Recall that the depth-2 circuit   $\Gamma$ accepts the set $X_2$ and  rejects the set $Y_2\subseteq Y_1$, assume $\Gamma$ is of form  $\bigwedge_{i=1}^{s} \Lambda_i$ where every $\Lambda_i$ is an $\vee$-clause. Thus, we have
$X_2 \subseteq \Lambda_i^{-1}(1)$ for every $i$,
and $Y_{2} \subseteq \bigcup_{i=1}^{s}\Lambda_i^{-1}(0)$, now choose $\Lambda_i$ so that $\Lambda_i$ maximizes  $Y_{2} \cap \Gamma_i^{-1}(0)$. After this, we set \label{Y3}$Y_3 = Y_{2} \cap \Lambda_i^{-1}(1)$ and $\Lambda:=\Lambda_i$. In summary, the $\vee$-clause  $\Lambda$ accepts the set $X_2$, rejects the set $Y_3 \subseteq Y_2$, the size of $Y_3$ is at least $|Y_2|/2^{-m^{1-\delta}} = 2^n \cdot 2^{-m^{1+\Theta(\frac{\log\log m}{\log m})}}$.

\paragraph{Reaching contradiction}~\\
Note that the $\vee$-clause $\Lambda$ accepts $X_2$ and  rejects $Y_3$, since $|Y_3| \geq 2^n \cdot 2^{-q}$, thus the number of input literals in $\Lambda$ is at most $q$, since every $y\in Y_3\subseteq Y_2$ is $q$-local limit for $X_2$, this means there exists $x\in X_2$ such that $\Lambda(x)=\Lambda(y)$ which is a contradiction as required.

\subsection{Proof of Lemma 3.1}\label{section3.1}
Set $\ell = \lfloor n/2 \rfloor + m$.  Recall that we set $X_1$ to be $\{x \in \{0,1\}^n\mid\|x\|_0 = \ell\}$ and the size of $X_1$ is $2^n \cdot 2^{-\Theta(\log m)}$. Denote the sphere with  center $x\in  \{0, 1\}^n$ and radius $r\in [n]$ by $\mathbb{S}_r(x) := \{y \in  \{0, 1\}^n : \|x-y\|_0 = r\}$ where
$\|x-y\|_0$ is the Hamming distance between $x$ and $y$.
Now we show there are many $x$'s in $X_1$ such that $\mathbb{S}_m(x)$ contains many $y$'s from $Y_1$.

\begin{claim} \label{claim3.4}
There is a subset $X_{1.1} \subseteq  X_1, \frac{|X_{1.1}|}{|X_1|} \geq  \frac{|Y_1|}{2|Y_0|}=2^{-m^{1-\delta}-1}$, such that for every $x \in  X_{1.1}$ we have
	\[
	\frac{|\mathbb{S}_m(x)\cap Y_1|}{\binom{|I_x|}{m}} \geq  \frac{1}{2}	\cdot \frac{|Y_1|}{|Y_0|} \geq 2^{-m^{1-\delta}-1}.
	\]
\end{claim}

\begin{proof} 
At first, note that for every $y\in \mathbb{S}_m(x)\cap Y_1$, since $\|x\|_0-\|y\|_0=\ell- \lfloor n/2\rfloor= m =\|x-y\|_0$, we must have $I_y\subset I_x$. Formally, if $I_y \not\subset I_x$, $|I_y \cap I_x|< \lfloor n/2\rfloor$, thus
	$ \|x-y\|_0 = |I_{x-y}| \geq |I_x \setminus I_y| \geq |I_x \setminus (I_y \cap I_x)| > m$ which is a contradiction. Now sample $\bm{x}$ uniformly random in $X_1$, then sample $\bm{y}$  uniformly random in $\mathbb{S}_m(\bm{x}) \cap Y_0$, note that $\bm{y}$ is uniformly distributed in $Y_0$. Now set $\alpha_x :={|\mathbb{S}_m(x)\cap Y_1|}/{\binom{|I_x|}{m}},\beta :={|Y_1|}/{|Y_0|}$, thus $\E_{\bm{x}\sim X_1}[\alpha_{\bm{x}}] =\Pr_{\bm{y}\sim Y_0}[\bm{y}\in Y_1]=\beta$. Furthermore, $\E_{\bm{x}\sim X_1}[\alpha_{\bm{x}}]\leq 
	\Pr_{\bm{x}\sim X_1}[\alpha_{\bm{x}} <\frac{\beta}{2} ]\cdot \frac{\beta}{2}  + \Pr_{\bm{x}\sim X_1}[\alpha_{\bm{x}} \geq\frac{\beta}{2}]\cdot 1$, thus $\Pr_{\bm{x}\sim X_1}[\alpha_{\bm{x}} \geq\frac{\beta}{2} ] \geq \beta/2$.
\end{proof}

\begin{lemma}\label{lem3.5} 
Given an element $x \in X_{1.1}$ such that ${|\mathbb{S}_m(x)\cap Y_1|}/\binom{|I_x|}{m} \geq  2^{-m^{1-\delta}-1}$.  Then there is a corresponding $x'\in \texttt{Maj}^{-1}(1)$ such that $Y_1-x'$ is $(m^{1+\Theta(\frac{\log\log m}{\log m})}/n, o(1))$-satisfying. Moreover, $\|x-x'\|_0 \leq m^{1-\delta}$.
\end{lemma}

Given an element $x \in X_{1.1}$, define $A := (\mathbb{S}_m(x)\cap Y_1)-x$, 
we know that $|A|/\binom{|I_x|}{m}\geq 2^{-m^{1-\delta}-1}$.  Since $|I_x| =\ell$, for simplicity and without loss of generality, assume $I_x =[\ell]$. Let
$\epsilon = \epsilon(m) := 1/ \log m = o(1)$ and $\kappa := \ell/m$. Now treat $A$ as a set family $A\subseteq \binom{I_x}{m}$, we want to force $A$ to be $\kappa^{1-\epsilon}$-spread by excluding a kernel. We isolate the kernel $I$ in a greedy way by set $I\subseteq [n]$ to be the largest set  where the $\kappa^{1-\epsilon}$- spreadness condition fails, that is 
\[
\Pr_{\bm{a}\sim A}[\bm{a}_I =\bm{1}_I]\geq \kappa^{-(1-\epsilon)|I|}.
\]
where $\bm{1}$ is the vector of all ones. Firstly, we want to make sure the kernel $I$ is not too large as claimed below.

\begin{claim}\label{claim3.6}
    $I\subseteq  I_x$ and $|I| \leq  m^{1-\delta}$. Furthermore, let  $a'$ be the indicator vector for $I$ in $[n]$, and $x'=x-a'$, the Hamming weight of $x'$ is strictly larger than $n/2$, thus $x' \in \texttt{Maj}^{-1}(1)$.
\end{claim}

\begin{proof}
At first, note that given any $J\subseteq [n]$ such that $J \setminus I_x \neq \emptyset$, the set $\{a_J =1_J|a\in A\}$ is empty, since $a_{[n]\setminus I_x}$ is the vector of all zeros. Thus $I$ must be a subset of $I_x$ and $|I|\leq m$. Now  treat every $a\in A$ as a set from the set family $\binom{[\ell]}{m}$, by the law of total probability,
	\begin{align*}
	\Pr_{\bm{a}\sim \binom{[\ell]}{m}}[\bm{a}_I =\bm{1}_I] 
	&= \Pr_{\bm{a}\sim \binom{[\ell]}{m}}[\bm{a}_I =\bm{1}_I|\bm{a}\in A]\cdot \Pr_{\bm{a}\sim \binom{[\ell]}{m}}[\bm{a}\in A]+\Pr_{\bm{a}\sim \binom{[\ell]}{m}}[\bm{a}_I =\bm{1}_I|\bm{a}\not\in A]\cdot \Pr_{\bm{a}\sim \binom{[\ell]}{m}}[\bm{a}\not\in A]\\
	&\geq \Pr_{\bm{a}\sim \binom{[\ell]}{m}}[\bm{a}_I =\bm{1}_I|\bm{a}\in A]\cdot \Pr_{\bm{a}\sim \binom{[\ell]}{m}}[\bm{a}\in A] \\
	&=\Pr_{\bm{a}\sim A}[\bm{a}_I =\bm{1}_I]\cdot \Pr_{\bm{a}\sim \binom{[\ell]}{m}}[\bm{a}\in A],
	\end{align*}
	thus we have 
    \begin{align*}
	    2^{-m^{1-\delta}-1}
        &\leq \Pr_{\bm{a}\sim \binom{[\ell]}{m}}[\bm{a}\in A] \le \frac{\Pr_{\bm{a}\sim \binom{[\ell]}{m}}[\bm{a}_I =\bm{1}_I]}{\Pr_{\bm{a}\sim A}[\bm{a}_I =\bm{1}_I]} \leq \frac{\binom{\ell-|I|}{m-|I|}/\binom{\ell}{m}}{\kappa^{-(1-\epsilon)|I|}}\\
        &= \frac{\prod_{i=0}^{|I|-1} \frac{m-i}{\ell-i}}{\kappa^{-(1-\epsilon)|I|}}\leq  \frac{\kappa^{-|I|}}{\kappa^{-(1-\epsilon)|I|}}= \kappa^{-\epsilon|I|}.
	\end{align*}
	Thus , $|I|\leq  (m^{1-\delta}+1)/(\epsilon \log\kappa)=(m^{1-\delta}+1)/(2-o(1))\leq  m^{1-\delta}$ where $\kappa:=\ell/m, \ell =\lfloor n/2 \rfloor+m,\epsilon=1/\log m$.
\end{proof}

Now we are ready to show after excluding the kernel $I$ out of $A$, we obtain  $A'$, and $A'$ is spread thus satisfying as claimed below.
\begin{claim}\label{claim3.7}
    Recall that $a'$ is the indicator vector  for $I$ in $[n]$ and  $x'=x-a'$. Define $A'$ to be the set $\{a\in A\mid a_I =\bm{1}_I\} -a'\subseteq (\mathbb{S}_m(x)\cap Y_1)-x' \subseteq Y_1-x'$, then $A'$ is $\kappa^{1-\epsilon}$-spread thus $(m^{1+\Theta(\frac{\log\log m}{\log m})}/n, o(1))$-satisfying, so is the set $Y_1 -x'$.
\end{claim}

\begin{proof}
Assume for contradiction that $A'$ is not $\kappa^{1-\epsilon}$-spread. This means there is a non-empty $J\subseteq [n]\setminus I$ such that 
	\[
	\Pr_{\bm{a}\sim A'}[\bm{a}_J =\bm{1}_J]> \kappa^{-(1-\epsilon)|J|}.
	\]
	We now show that $\Pr_{\bm{a}\sim A}[\bm{a}_{I\cup J} =\bm{1}_{I\cup J}]> \kappa^{-(1-\epsilon)|{I\cup J}|}$ , which would contradict that $I$ is maximal such set.
	\begin{align*}
	\Pr_{\bm{a}\sim A}[\bm{a}_{I\cup J} =\bm{1}_{I\cup J}] 
	&= \Pr_{\bm{a}\sim A}[\bm{a}_{I} =\bm{1}_{I}] \cdot \Pr_{\bm{a}\sim A}[\bm{a}_{J} =\bm{1}_{J}|\bm{a}_{I} =\bm{1}_{I}] \\
	&= \Pr_{\bm{a}\sim A}[\bm{a}_{I} =\bm{1}_{I}] \cdot \Pr_{\bm{a}\sim A'}[\bm{a}_{J} =\bm{1}_{J}] \\
	&>  \kappa^{-(1-\epsilon)|I|} \cdot  \kappa^{-(1-\epsilon)|J|} \\
	&\geq  \kappa^{-(1-\epsilon)|{I\cup J}|}
	\end{align*}
Finally, we verify the parameters. Since $A'\subseteq \binom{[n]}{m-|I|}$ is $\kappa^{1-\epsilon}$-spread, by Lemma \ref{lem2.3}, $A'$ is $(p,\epsilon)$ satisfying where  $p =C/\kappa^{1-\epsilon} \cdot \log((m-|I|)/\epsilon) =\Theta(m/n)\cdot \Theta(\log m)=m^{1+\Theta(\frac{\log\log m}{\log m})}/n$ as required.
\end{proof}

Note that Claim \ref{claim3.6} and Claim \ref{claim3.7} together imply Lemma \ref{lem3.5}. Now we can use Lemma \ref{lem3.5}  and Claim \ref{claim3.4} to prove Lemma \ref{lem3.1} as follows.

\begin{proof}[Proof of  Lemma \ref{lem3.1}]
We define our mirror set to be $M := \{x': x \in X_{1.1}\} \subseteq \texttt{Maj}^{-1}(1)$ where we generate each $x'$ by applying
Lemma \ref{lem3.5} to each $x \in X_{1.1}$ of Claim \ref{claim3.4}. Each $x' \in M$  could be generated from any $x \in X_{1.1}$ with $\|x-x'\|_0 \leq m^{1-\delta}$. Hence the size of  $M$ is at least
\begin{align*}
\frac{|X_{1.1}|}{\sum_{i=0}^{m^{1-\delta}}\binom{|I_x|}{i}}&\geq |X_{1.1}| \cdot 2^{-m^{1-\Theta(\delta)+\Theta(\frac{\log\log m}{\log m})}} \geq |X_{1.1}| \cdot 2^{-m^{1-\Theta(\delta)}} \\
&\geq 2^n\cdot 2^{-m^{1-\delta}-\Theta(\log m)}\cdot 2^{-m^{1-\Theta(\delta)}}  \geq 2^n\cdot 2^{-m^{1-\Theta(\delta)}}
\end{align*}
since $\delta$ is $\omega (\frac{\log\log m}{\log m})$.
\end{proof}

\section{Conclusion and discussion}\label{section4}

We provide a  proof of top-down depth-four lower bound for Majority function by constructing a simple mirror set. This is done by avoiding the flipping in the case of Parity function. But this case-by-case analysis for different functions is still unsatisfying. The first question is that can we provide a  proof for general case unifying the cases like Parity, Majority and potentially all  functions with large average sensitivity? Another interesting direction is to consider replacing the gates in the bottom layer to modular gates.  Riazanov, Sofronova, and Sokolov \cite{DBLP:conf/innovations/RiazanovS026} recently prove lower bounds for circuits of type $ \vee \circ (\wedge \circ \vee \circ A)$ where $A$ is a affine map.  This type of circuit is just the disjunction of composition of a CNF and an affine map.  They manage to prove such lower bound for the middle slice function. So maybe it's possible to extend our results to depth-four circuits composed with  an affine map at the bottom. Finally, the ultimate question is that can we prove a top-down lower bound beyond depth-four? The main obstacle of the current approach is the way of constructing the mirror set. We construct current mirror set out of $f^{-1}(1)$ which allows us to use every element in $f^{-1}(1)$, it is not clear how to do this if we are asked to construct a mirror set out of some strict subset of $f^{-1}(1)$, perhaps new ideas are needed.

\bibliographystyle{alpha}
\bibliography{hw}

\end{document}

%% file: mycmd.tex
\newtheorem{theorem}{Theorem}[section]

\newtheorem{lemma}[theorem]{Lemma}

\theoremstyle{definition}
\newtheorem{defn}[theorem]{Definition}

\newtheorem{claim}[theorem]{Claim}

\numberwithin{equation}{section}

\DeclareMathOperator{\E}{\mathbb{E}}